\documentclass{amsart}
\usepackage{amssymb, amsmath, epsfig}
\usepackage{amscd}

\newtheorem{thm}{Theorem}
\newtheorem{prop}[thm]{Proposition}

\newtheorem{remark}{\it Remark}
\newtheorem{example}{Example}

\newcommand{\R}{\mathbb{ R}}

\newcommand{\g}{\mathfrak{ g}}

\DeclareMathOperator{\pr}{pr} \DeclareMathOperator{\Span}{span}
\DeclareMathOperator{\diag}{diag} \DeclareMathOperator{\ad}{ad}
\DeclareMathOperator{\Ad}{Ad} \DeclareMathOperator{\tr}{tr}

\begin{document}


\title[Invariant measures of modified LR and L+R systems]{Invariant measures of modified LR and L+R systems}
\author{Bo\v zidar Jovanovi\' c}

\baselineskip=14pt

\maketitle

\centerline{\sc \small Mathematical Institute SANU}
\centerline{\sc \small Kneza Mihaila 36, 11000, Belgrade, Serbia}
\centerline{\small {\it e--mail}: bozaj@mi.sanu.ac.rs}

\begin{abstract}
We introduce a class of dynamical systems having an invariant
measure, the modifications of well known systems on Lie groups: LR
and L+R systems. As an example, we study modified Veselova
nonholonomic rigid body problem, considered as a dynamical system
on the product of the Lie algebra $so(n)$ with the Stiefel variety
$V_{n,r}$, as well as the associated $\epsilon$L+R system on
$so(n)\times V_{n,r}$. In the 3--dimensional case, these systems
model the nonholonomic problems of a motion of a ball and a rubber
ball over a fixed sphere.
\end{abstract}

\section{Introduction}

This paper describes a class of dynamical systems allowing an
invariant measure, a modification of LR systems introduced by
Veselov and Veselova \cite{VeVe2} and L+R systems introduced by
Fedorov (see \cite{FeRCD, FeKo}). In particular, they model the
nonholonomic problems of rolling the Chaplygin ball and the rubber
Chaplygin ball over a spherical surface.

Recall that the motion of the Chaplygin (balanced, dynamically
asymmetric) ball over a fixed spherical surface is described by
the equations
\begin{equation}
\frac{d}{dt}{\vec{\mathbf k}}=\vec{\mathbf k}\times\vec \omega,
\qquad \frac{d}{dt}{\vec \gamma}=\epsilon \vec
\gamma\times\vec\omega, \label{Chap}\end{equation} where
$\vec\omega$ is the angular velocity of the ball, $\vec\gamma$ is
the unit vector directed from the center of the fixed sphere to
the point of contact, $\vec{\mathbf k}=\mathbb I \vec\omega+ D\vec
\omega-D (\vec \omega,\vec\gamma)\vec \gamma$ is the momentum of
the ball with respect the contact point, and $\mathbb
I=\diag(I_1,I_2,I_3)$ is the inertia operator of the ball (e.g.,
see \cite{BM, BFM}).

The parameters $D$ and $\epsilon$ are equal to $m^2\rho$ and
${\sigma}/({\sigma\pm\rho})$, where $m$ is the mass of the rolling
ball, $\rho$ its radius and $\sigma$ is the radius of the fixed
sphere. The sign "$+$" denotes the rolling over outer surface of
the fixed sphere, while the sign "$-$" denotes either the rolling
over inner surface of the fixed sphere ($\sigma>\rho$), or the
case where the fixed sphere is within the rolling ball
($\sigma<\rho$, in this case the rolling ball is actually a
spherical shell). As $\sigma$ tends to infinity, $\epsilon$ tends
to $1$, and we obtain the equation of the rolling of the Chaplygin
ball over a horizontal plane.

In the space $(\vec{\omega}, \vec\gamma)$, the density $\mu$
 of an invariant measure is equal to
\begin{equation}
\mu=\sqrt{ \det(\mathbb I+D\mathbb E)\left(1-D(\vec\gamma,(\mathbb
I+D\mathbb E)^{-1}\vec\gamma )\right)}, \label{mu-ch}
\end{equation} the expression given by Chaplygin for $\epsilon=1$ \cite{Ch1}, and by Yaroshchuk for $\epsilon \ne 1$ \cite{Ya}. Through the paper
the operator $\mathbb E$ denotes the identity operator on the
appropriate spaces.

Similarly, the motion of the rubber Chaplygin's ball over a fixed
spherical surface is described by the equations (Ehlers and
Koiller \cite{EK}, see also Borisov and Mamaev \cite{BM2, BM})
\begin{equation}
\frac{d}{dt}{\vec m}=\vec m\times \vec \omega+\lambda\vec\gamma,
\qquad \frac{d}{dt}{\vec\gamma}=\epsilon
\vec\gamma\times\vec\omega, \qquad (\vec\omega,\vec\gamma)=0,
\label{rubber}\end{equation} where $ \vec m =(\mathbb I+D\mathbb
E)\vec\omega=\mathbf I\vec\omega$, $\lambda=(\vec m,\mathbf
I^{-1}\vec\gamma)/(\vec\gamma,\mathbf I^{-1}\vec\gamma)$. Here the
constraint $(\vec\omega,\vec\gamma)=0$ model the "rubber" property
of the ball: the rotations of the ball around the normal to the
spherical surface at the contact point are forbidden.

Note that the above system is well defined on the whole phase
space $(\vec{m}, \vec\gamma)$, and $(\vec\omega,\vec\gamma)=const$
denotes its first integral. In the space $(\vec{m}, \vec\gamma)$,
or $(\vec{\omega},\vec{\gamma})$, the system has an invariant
measure with the density $\mu_\epsilon$ given by
\begin{equation}\label{mu-rubber}
\mu_\epsilon=(\mathbf
I^{-1}\vec\gamma,\vec\gamma)^\frac{1}{2\epsilon},
\end{equation}
see \cite{VeVe2, EKR} for $\epsilon=1$, and \cite{EK} for
$\epsilon \ne 1$.

The existence of an invariant measure for some nonholonomic
systems is an important property related to the geometry of the
problem, its Hamiltonization, as well as to the possibility to
solve the problem by quadratures (e.g., see \cite{AKN, BN, BBM,
EKR, FeJo, FNM, Ko, ZB}).

We shall prove that the densities \eqref{mu-ch} and
\eqref{mu-rubber} are particular cases of the densities of
invariant measures of $\epsilon$--modified L+R and LR systems,
respectively (see Theorems 1, 3 and 5, and Examples 1 and 2).
Further, as a specific example, we give the expression of an
invariant measure for the modified Veselova nonholonomic rigid
body problem (Theorem 4), considered as a dynamical system on the
product of the Lie algebra $so(n)$ with the Stiefel variety
$V_{n,r}$, as well as the expression of an invariant measure for
the associated $\epsilon$L+R system on $so(n)\times V_{n,r}$
(Theorem 6). For $r=1$, these systems represent natural
multidimensional generalizations of the equations \eqref{Chap} and
\eqref{rubber}.

\section{LR and L+R systems}

\subsection {LR systems}
Let $G$ be a $n$--dimensional compact connected Lie group $G$,
$\mathfrak{g}=Lie(G)$ its Lie algebra, $\langle \cdot,\cdot
\rangle$ an $\Ad_G$-invariant scalar product on $\mathfrak{g}$,
and let $(\cdot,\cdot)_\mathbb I$ be a left-invariant metric on
$G$ given by the positive definite operator (called inertia
operator) ${\mathbb I}: \mathfrak{g} \to \mathfrak{g}\cong
\mathfrak g^*$: $(\eta_1,\eta_2)_\mathbb I=\langle {\mathbb
I}(g^{-1}\eta_1),g^{-1}\eta_2\rangle$ $\eta_1,\eta_2\in T_g G$.
Here we identify $\mathfrak g$ and $\mathfrak g^*$ by $\langle
\cdot,\cdot \rangle$.

The \textit{LR system} on $G$ is a nonholonomic Lagrangian system
$(G,L,\mathcal D)$, where $L=\frac12(\dot g,\dot g)_\mathbb I$ is
a left-invariant Lagrangian  and $\mathcal D$ is a right-invariant
nonintegrable distribution on the tangent bundle $TG$, determined
by its restriction $\mathfrak{d}$ to the Lie algebra \cite{VeVe2}.
Let ${\mathfrak h}$ be the orthogonal complement of $\mathfrak{d}$
with respect to $\langle \cdot,\cdot \rangle$. Then the
right-invariant constraints, in the left-trivialization of $TG$,
can be written as $\omega\in \Ad_{g^{-1}}\mathfrak d$, or $\langle
\omega, \Ad_{g^{-1}} \mathfrak{h}\rangle=0$, where
$\omega=g^{-1}\cdot \dot g$ is the angular velocity.

The LR system $(G,L,\mathcal D)$ is described by the
Euler--Poincar\'e--Chetayev equations on $T^*G(m,g)$ (or
$TG(\omega,g)$)
\begin{equation} \label{EPS}
\dot m = [m,\omega]+\sum^{k}_{s=1}\lambda^{s}\, e_s\ , \qquad \dot
g = g \omega,
\end{equation}
where $m=\mathbb I\omega\in\mathfrak g^*\cong \mathfrak g$ is the
angular momentum, $e_s=g^{-1}\cdot E_s\cdot g$, $E_1,\dots,E_k$ is
the orthonormal base of $\mathfrak h$, and $\lambda^{s}$ are
Lagrange multipliers which can be found by differentiating the
constraints $\phi_s=\langle \omega, e_s \rangle=0$, $s=1,\dots,k$.
These equations define a dynamical system on the whole cotangent
bundle $T^*G$, and the right-invariant constraint functions
$\phi_s$
 are its first integrals.
Also, to the system (\ref{EPS}) we can associate a closed system
on  $\mathfrak g^{m+1}(m, {e}_1,\dots, {e}_k)$:
\begin{equation} \label{EPS1}
\dot m
 = [m,\omega]+\sum^{k}_{s=1}\lambda^{s}{e}_s, \qquad \dot {e}_s  = [e_s,\omega].
\end{equation}

 The LR system (\ref{EPS1}) possesses an invariant measure
$\mu\,{\mathrm d}m\wedge {\mathrm d}e_1 \wedge \dots \wedge
{\mathrm d}e_k$
 with density (see \cite{VeVe2})
\begin{equation} \label{measure1}
\mu=\sqrt{ \det \langle{e_s},{\mathbb I}^{-1} {e_l}\rangle },
\qquad s,l=1,\dots, k,
\end{equation}
implying that the original system \eqref{EPS} has an invariant
measure $\sqrt{ \det ( {\mathbb I}^{-1} |_{g^{-1} \mathfrak{h} g}
) }\,\varOmega^n$. Here ${\mathbb I}^{-1}|_{g^{-1} \mathfrak{h}
g}$ is the restriction of the inverse inertia tensor to the linear
space $g^{-1}\mathfrak{h} g \subset {\mathfrak g}$, by  $\mathrm d
m$, $\mathrm d e_i$ we denoted the standard measures on $\g$ with
respect to the metric $\langle\cdot,\cdot\rangle$, and $\varOmega$
is the standard symplectic form on $T^*G$.

\subsection{L+R systems}
In addition to the inertia operator ${\mathbb I}$ defining the
left-invariant metric $(\cdot,\cdot)_{{\mathbb I}}$, introduce a
constant linear operator $\Pi^{0}: \g\rightarrow \g$ defining a
right-invariant scalar product $(\cdot,\cdot)_{\Pi}$ on $G$:
$$
(\eta_1,\eta_2)_{\Pi }=\langle \Pi ^{0} \eta_1 g^{-1}, \eta_2
g^{-1}\rangle=\langle \Pi_g \eta_1 g^{-1}, \eta_2 g^{-1}\rangle,
\qquad \Pi_g=\Ad_{g^{-1}}\circ \Pi^0\circ\Ad_g,
$$
$\eta_1,\eta_2\in T_{g} G$. We suppose that ${\kappa}_g= {\mathbb
I}+\Pi_g$ is nondegenerate and positive definite on the whole
group $G$. The {\it L+R system} on $G$ is defined as a dynamical
system
\begin{equation}
\frac{d}{dt}\left( \mathbb I\omega+\Pi_g\omega\right)= [\mathbb
I\omega+\Pi_g\omega, \omega], \qquad \dot g=g\cdot \omega.
\label{L+R}
\end{equation}
This is the the modification of the geodesic motion on the group
$G$ with respect to the metric $(\cdot,\cdot)_\mathbb
I+(\cdot,\cdot)_\Pi$, by rejecting the term $[\omega,\Pi_g\omega]$
at the right hand side of the first equation in \eqref{L+R}.

In the view of the definition of $\Pi_g$, its evolution is given
by the  matrix equation $\dot\Pi_g=\Pi_g\circ
\ad_\omega+\ad_\omega^T\circ\Pi_g$. Note that for compact group we
have $\ad_\omega^T=-\ad_\omega$, and therefore, we get a closed
system
\begin{equation}
\frac{d}{dt}\left( \mathbb I\omega+\Pi\omega\right)= [\mathbb
I\omega+\Pi\omega, \omega], \qquad \dot\Pi=[\Pi,\ad_\omega]
\label{L+R2}
\end{equation}
on $\mathfrak g\times {\bf Sym}(\mathfrak g)$, where ${\bf
Sym}(\mathfrak g)$ is the space of of symmetric operators on $\g$,
which we also refer as a L+R system. It possesses the kinetic
energy integral $\frac 12\langle\mathbb
I\omega+\Pi\omega,\omega\rangle$ and an invariant measure
$\mu\,{\mathrm d}\omega \wedge {\mathrm d}\Pi$
 with density
\begin{equation*}
\mu=\sqrt{\det ({\mathbb I}+\Pi)},
 \end{equation*} where $\mathrm
d\Pi$ is the standard measure on  ${\bf Sym}(\g)$ \cite{FeRCD}. It
appears that every L+R system can be seen as a reduced system of a
certain LR system on a direct product $G\times \g^n$, where $\g$
is considered as a commutative group (see Theorem 3.3 in
\cite{Jo3}).


\section{Modified LR systems}

As the sphere--sphere problems \eqref{Chap}, \eqref{rubber}
suggest, it is natural to consider modifications of the equations
\eqref{EPS1} and \eqref{L+R2}. We define the {\it $\epsilon$LR
system} on the space
$$
\mathfrak g^{k+1}(\omega,e_1,\dots,e_k),\qquad \text{or} \qquad
\g^{k+1}(m,e_1,\dots,e_k),
$$
by the equations
\begin{eqnarray} \label{mLR}
\dot m &=&   [m,\omega]+\Lambda, \qquad
\Lambda=\sum^{k}_{i=1}\lambda^{i}\, e_i, \qquad
m=\mathbb I\,\omega, \\
 \dot e_i &=& \epsilon [e_i,\omega],\label{mLR2}
\end{eqnarray}
$i=1,\dots,k$.  The term $\Lambda$ can be interpreted as a
reaction force and it is determined from the condition that the
trajectories $(\omega(t),e_1(t),\dots,e_k(t))$ satisfy constraints
\begin{equation}\label{veze}
\phi_i=\langle \omega, e_i \rangle=c_i=const, \qquad i=1,\dots,k.
\end{equation}
By differentiating \eqref{veze}, we obtain the linear system
\begin{equation}\label{racun}
\epsilon \langle [e_i,\omega],\omega\rangle+\langle e_i,\mathbb
I^{-1}[m,\omega]\rangle+\sum_j\lambda^j\mathbb A_{ij}=0, \qquad
\mathbb A_{ij}=\langle e_i, \mathbb I^{-1} e_j\rangle,
\end{equation}
and we get the Lagrange multipliers in the form
\begin{equation}\label{m-lambde}
\lambda^i=-\sum_j \langle e_j,\mathbb I^{-1}[m,\omega]\rangle
\mathbb A^{ij}.
\end{equation}
Here $\mathbb A^{ij}$ is the inverse of the matrix $\mathbb
A_{ij}=\langle \mathbb I^{-1} e_i,e_j\rangle$, $i,j=1,\dots,k$. In
particular, $\Lambda$ does not depend on $\epsilon$.

\begin{thm}\label{TLR}
The $\epsilon$LR system \eqref{mLR}, \eqref{mLR2},
\eqref{m-lambde} has an invariant measure
\begin{equation}\label{emera}
\mu_\epsilon\,\mathrm d m\wedge {\mathrm d}e_1 \wedge \dots \wedge
{\mathrm d}e_k, \quad \text{i.e.,} \quad \mu_\epsilon\,\mathrm
d\omega\wedge {\mathrm d}e_1 \wedge \dots \wedge {\mathrm d}e_k
\end{equation}
where the density is given by
\begin{equation}\label{gustina}
\mu_\epsilon=(\Delta)^{\frac{1}{2\epsilon}}\, ,
\qquad\Delta=\det(\mathbb A_{ij})=\det(\langle \mathbb I^{-1}
e_i,e_j\rangle)\, , \qquad i,j=1,\dots,k.
\end{equation}
\end{thm}

\begin{proof}
Consider the vector field
$$
X=(\dot m,\dot e_1,\dots,\dot
e_k)=([m,\omega]+\Lambda,\epsilon[e_1,\omega],\dots,\epsilon[e_k,\omega]).
$$
The
volume form \eqref{emera} is invariant with respect to the flow of
\eqref{mLR} if and only if
$$
\mathcal L_X\,\mu_\epsilon\,\mathrm d m\wedge {\mathrm d}e_1
\wedge \dots \wedge {\mathrm d}e_k=0,
$$
i.e., if $\mu_\epsilon$ satisfies the Liouville equation
\begin{equation}\label{uslov}
\mu_\epsilon\left(\mathrm{div}(\dot m)+\mathrm{div}(\dot
e_1)+\dots+\mathrm{div}(\dot e_k)\right)+\dot\mu_\epsilon=0,
\end{equation}
where $\mathrm{div}(\cdot)$ is the standard divergence on
$\mathfrak g$. It is clear that $\mathrm{div}(\dot e_i)=0$,
$i=1,\dots,k$. In Theorem 1 \cite{VeVe2},  it is proved that
\begin{equation}
\mathrm{div}(\dot m)=\mathrm{div}\left(\Lambda\right)=
-\sum_{i,j=1}^k \mathbb A^{ij}\langle [\mathbb I^{-1} e_i,\mathbb
I^{-1} m],e_j\rangle . \label{div}\end{equation}

For any regular matrix $\mathbb A$ we have the identity
\begin{equation}\label{ID}
\frac{d}{dt}\det \mathbb A=\det\mathbb A\tr(\mathbb
A^{-1}\dot{\mathbb A}).
\end{equation}
Thus,
\begin{eqnarray*}
\dot\mu_\epsilon&=&\frac{1}{2\epsilon}(\Delta)^{\frac{1}{2\epsilon}-1}\dot\Delta=
\frac{1}{2\epsilon}\mu_\epsilon\sum_{i,j=1}^k \mathbb A^{ij}\frac{d}{dt}{\mathbb A_{ji}}\\
&=&\frac{1}{2}\mu_\epsilon\sum_{i,j=1}^k \mathbb
A^{ij}\left(\langle\mathbb I^{-1}
e_i,[e_j,\omega]\rangle+\langle\mathbb I^{-1}
e_j,[e_i,\omega]\rangle\right)\\
&=&\mu_\epsilon\sum_{i,j=1}^k \mathbb A^{ij}\langle\mathbb I^{-1}
e_i,[e_j,\omega]\rangle,
\end{eqnarray*}
which together with \eqref{div} implies \eqref{uslov}. Since
$\mathrm d m=\det\mathbb I\cdot \mathrm d\omega=const\cdot\mathrm
d\omega$, if instead of $m$ we take the variable $\omega$, the
expression for an invariant measure remains the same.
\end{proof}

\subsection{Momentum equation}
As above, let $E_1,\dots,E_k$ be an orthonormal base of $\mathfrak
h$. Further, let $E_{k+1},\dots,E_n$ be an orthonormal base of
$\mathfrak d$ and $\mathcal O_{E_i}$ be the adjoint orbit of
$E_i$, $i=1,\dots,n$. It is clear that we can also consider the
$\epsilon$LR system \eqref{mLR}, \eqref{mLR2}, \eqref{m-lambde} on
the space
$$
\mathcal M=\{(\omega,e_1,\dots,e_n)\in \mathfrak g\times \mathcal
O_{E_1}\times \cdots \times \mathcal O_{E_n}\, \vert\, \langle
e_i,e_j \rangle=\delta_{ij}\},
$$
simply by taking $i=1,\dots,n$ in the equation \eqref{mLR2}. Then
it has an invariant measure
\begin{equation}\label{meraLR}
\mu_\epsilon\,\mathrm{d}\omega\wedge{\mathrm d}e_1 \wedge \dots
\wedge {\mathrm d}e_n \vert_\mathcal M.
\end{equation}
Moreover, on $\mathcal M$ we have well defined orthogonal
projections $\pr_\mathcal H$, $\pr_\mathcal D$ ($\pr_\mathcal
H+\pr_\mathcal D=\mathbb E$) onto the complementary subspaces of
$\g$:
$$
\mathcal H=\Span\{e_1,\dots,e_k\}, \qquad \mathcal
D=\Span\{e_{k+1},\dots,e_n\},
$$
which, according to \eqref{mLR2}, satisfy the equations
\begin{equation}\label{projektori}
\frac{d}{dt}\pr_{\mathcal D}=\epsilon [\pr_{\mathcal
D},\ad_\omega], \qquad \frac{d}{dt}\pr_{\mathcal H}=\epsilon
[\pr_{\mathcal H},\ad_\omega].
\end{equation}

Note that  with the above notation, we can express $\Lambda$ in
\eqref{mLR} as
$$
\Lambda=-\mathbb A^{-1}\pr_{\mathcal H} \mathbb I^{-1} [\mathbb
I\omega,\omega],
$$
where $\mathbb A^{-1}: \mathcal H\to \mathcal H$ is the inverse of
the mapping $\mathbb A=\mathbb I^{-1}\vert_\mathcal H=\pr_\mathcal
H\circ \mathbb I^{-1}: \mathcal H \to \mathcal H. $

Following \cite{FeJo}, let us introduce the momentum
$$
\mathbf m=\pr_{\mathcal D} \mathbb I\omega+\pr_{\mathcal
H}\omega=\mathbf J\omega, \qquad \mathbf J=\pr_{\mathcal D}
\mathbb I+\pr_{\mathcal H}=\mathbb E+\pr_{\mathcal D}(\mathbb
I-\mathbb E).
$$

Consider the space
$$
\mathcal N=\{(\mathbf m,e_{k+1},\dots,e_n)\in \mathfrak g\times
\mathcal O_{E_{k+1}}\times \cdots \times \mathcal O_{E_n}\,
\vert\, \langle e_i,e_j \rangle=\delta_{ij}\}.
$$

\begin{prop}
The $\epsilon$LR system on $\mathcal N$ has the following form
\begin{eqnarray}\label{mom-eq}
&&\dot{\mathbf m}=\epsilon [\mathbf
m,\omega]+(1-\epsilon)\pr_{\mathcal D}[\mathbb I\omega,\omega], \\
\nonumber && \dot e_i=\epsilon [e_i,\omega], \qquad i=k+1,\dots,n.
\end{eqnarray}
\end{prop}

\begin{proof}
From the equations \eqref{mLR} and \eqref{projektori}, the
evolution of the momentum $\mathbf m$ reads
\begin{eqnarray*}
\dot{\mathbf m} &=& \frac{d}{dt}\left(\pr_{\mathcal D} \mathbb
I\omega+\pr_{\mathcal
H}\omega\right)\\
&=&\epsilon\left(\pr_{\mathcal D}[\omega,\mathbb
I\omega]-[\omega,\pr_{\mathcal D} \mathbb
I\omega]\right)+\pr_{\mathcal D}\left([\mathbb
I\omega,\omega]+\Lambda\right)\\
&&\quad +\epsilon\left(\pr_{\mathcal
H}[\omega,\omega]-[\omega,\pr_{\mathcal H}
\omega]\right)+\pr_{\mathcal H}\dot\omega\\
&=& \epsilon [\mathbf m,\omega]+(1-\epsilon)\pr_{\mathcal
D}[\mathbb I\omega,\omega]+\pr_{\mathcal H}\dot\omega\,.
\end{eqnarray*}

On the other hand, from \eqref{veze} it follows: $\epsilon \langle
[e_i,\omega],\omega\rangle+\langle e_i,\dot\omega\rangle=\langle
e_i,\dot\omega\rangle=0$, $i=1,\dots,k$. Whence
\begin{equation}\label{HO}
\pr_{\mathcal H}\dot\omega =\sum_{i=1}^k \langle
\dot\omega,e_i\rangle e_i=0.
\end{equation}

Finally note, since the linear operator $\mathbf J$ is invertible,
$\omega=\mathbf J^{-1}\mathbf m$ and we have a closed system
\eqref{mom-eq} on $\mathcal N$.
\end{proof}

In the case $\epsilon=1$, the first equation in \eqref{mom-eq}
reduces to the momentum equation (2.7) of \cite{FeJo}. Also, as in
Theorem 3.2 \cite{FeJo}, we have
\begin{equation}\label{det}
\det\mathbb I\cdot \det\mathbb A=\det\mathbb I\cdot \det(\mathbb
I^{-1}\vert_\mathcal H)=\det\mathbf J=\det(\mathbb I\vert_\mathcal
D),
\end{equation}
where $\mathbb I\vert_\mathcal D=\pr_\mathcal D\circ \mathbb I: \,
\mathcal D \to \mathcal D$. Therefore:
\begin{eqnarray*}
\mu_\epsilon\, {\mathrm d}\omega &=&
\mu_\epsilon\det\frac{\partial\omega}{\partial\mathbf m}\,{\mathrm
d}\mathbf m\\
&=&\left(\det(\mathbb I^{-1}\vert_\mathcal
H)\right)^{\frac{1}{2\epsilon}}(\det\mathbf J)^{-1}{\mathrm
d}{\mathbf m}\\
&=&(\det\mathbb I)^{-\frac{1}{2\epsilon}}(\det\mathbb
I\vert_\mathcal D)^{\frac{1}{2\epsilon}-1}{\mathrm d}\mathbf m\,.
\end{eqnarray*}

Combining the expression of the invariant measure \eqref{meraLR}
on the space $\mathcal M$ and the above identity, we get the
following statement:

\begin{thm}\label{TLR2}
The $\epsilon$LR system \eqref{mom-eq} has an invariant measure
\begin{equation*}
\tilde \mu_\epsilon\,{\mathrm d}{\mathbf m}\wedge {\mathrm
d}e_{k+1} \wedge \dots \wedge {\mathrm d}e_n\vert_\mathcal N,
\end{equation*}
where the density is given by
\begin{equation*}
\tilde{\mu}_\epsilon=(\det\mathbb I\vert_\mathcal
D)^{\frac{1}{2\epsilon}-1}=(\det(\langle \mathbb I\,
e_i,e_j\rangle))^{\frac{1}{2\epsilon}-1}, \qquad i,j=k+1,\dots,n.
\end{equation*}
\end{thm}

\begin{example}{\rm Let us consider
the case when $\mathcal H$ is the isotropy algebra of
$\gamma=e_1$: $\mathcal H=\{x\in\g\,\vert\,[x,\gamma]=0\}$. Then
$\mathcal D$ can be identified with the tangent plane $T_\gamma
\mathcal O$ of the adjoint orbit $\mathcal O$ through $\gamma$.
Since $\mathcal H$ and $\mathcal D$ are uniquely determined by
$\gamma$, we can write the closed system in variables $(\mathbf
m,\gamma)$ or $(\omega,\gamma)$:
\begin{eqnarray}\label{mom-eq*}
&&\dot{\mathbf m}=\epsilon [\mathbf
m,\omega]+(1-\epsilon)\pr_{\mathcal D}[\mathbb
I\omega,\omega],\qquad \dot \gamma=\epsilon [\gamma,\omega],\\
&&\nonumber \mathbf m=\pr_{\mathcal D} \mathbb
I\omega+\pr_{\mathcal H}\omega.
\end{eqnarray}
 Also, in the special case when $\mathcal H$ is
one dimensional, spanned by $\gamma=e_1$, we have the equations
\begin{eqnarray}\label{mLR*}
&&\dot m = [m,\omega]+\lambda\gamma,\qquad
 \dot \gamma =\epsilon [\gamma,\omega],\\
 &&\nonumber\lambda=-{\langle \gamma,\mathbb I^{-1}[m,\omega]\rangle}/{\langle
\mathbb I^{-1}\gamma,\gamma\rangle},\qquad m=\mathbb I\omega.
\end{eqnarray}

The above examples coincide in the case of the Lie algebra
$so(3)$. Under the usual isomorphism between the Euclidian space
${\mathbb R}^3$ and $so(3)$
\begin{equation}
\vec X=(X_1,X_2,X_3)\longmapsto X=\left(\begin{matrix}
0 & -X_3 & X_2 \\
X_3 & 0 & -X_1 \\
-X_2 & X_1 & 0
\end{matrix}\right), \label{iso}
\end{equation}
replacing the inertia operator $\mathbb I$ by $\mathbf I=\mathbb
I+D\mathbb E$, the equations \eqref{mLR*} recover the equations
\eqref{rubber} of the rubber Chaplygin ball, while the expression
for the density of the invariant measure \eqref{gustina} recovers
the density \eqref{mu-rubber}. Also, according to \eqref{mom-eq*},
we can write the rubber Chaplygin ball equations in the equivalent
form
\begin{eqnarray}\label{mom-eq**}
\frac{d}{dt}{\vec{\mathbf m}}&=&\epsilon \vec{\mathbf m}
\times\vec\omega+(1-\epsilon)\left(\mathbf
I\vec\omega\times\vec\omega-(\mathbf I\vec\omega\times\vec\omega,\vec\gamma)\vec\gamma\right),\qquad \frac{d}{dt}{\vec\gamma}=\epsilon \vec\gamma\times\vec\omega,\\
\nonumber \vec{\mathbf m}&=& \mathbf
I\vec\omega+(\vec\gamma,\omega-\mathbf I\vec\omega)\vec\gamma.
\end{eqnarray}
}\end{example}

\subsection{Modified Veselova problem}\label{MVS} Following \cite{FeJo}, we define $\epsilon$--modified
Veselova nonholonomic rigid body problem as follows. Let $\mathbf
e_1,\dots,\mathbf e_n$ be a (moving) orthonormal frame of the
Euclidean space $(\mathbb R^n, (\cdot,\cdot))$. Consider the
orthogonal decomposition
\begin{equation}\label{ort-so(n)}
so(n)=\mathcal H_r\oplus\mathcal D_r,
\end{equation}
\begin{equation*}
\mathcal H_r=\Span\{\mathbf e_p\wedge \mathbf e_q\,\vert\,r<p<q\le
n\}, \,\, \mathcal D_r=\Span\{\mathbf e_i\wedge \mathbf
e_j\,\vert\,1\le i\le r, 1\le j \le n\}.
\end{equation*}
 Then
$$
\pr_{\mathcal D_r}\eta=\Gamma \eta+\eta\Gamma-\Gamma\eta\Gamma,
\qquad \eta\in so(n),
$$
where
$$
\Gamma=\mathbf e_1\otimes \mathbf e_1+\dots+\mathbf e_r\otimes
\mathbf e_r.
$$

The above projection depends only on the point $(\mathbf
e_1,\dots,\mathbf e_r)$ of the Stiefel variety $V_{n,r}$, realized
as a submanifold of $\mathbb R^{nr}(\mathbf e_1,\dots,\mathbf
e_r)$ by the constraints
$$
V_{n,r}: \qquad (\mathbf e_i,\mathbf e_j)=\delta_{ij}, \qquad 1\le
i,j \le r.
$$
Therefore, instead of using the variables $\mathbf m, \mathbf
e_i\wedge \mathbf e_j$, $1\le i\le r, 1\le j \le n$, we can write
the equations \eqref{mom-eq} on
$$
so(n)\times V_{n,r}\,(\mathbf m,\mathbf e_1,\dots,\mathbf e_r)
$$
as
follows:
\begin{eqnarray}
\nonumber&&\dot{\mathbf m}=\epsilon [\mathbf
m,\omega]+(1-\epsilon)\left(\Gamma[\mathbb
I\omega,\omega]+[\mathbb I\omega,\omega]\Gamma-\Gamma[\mathbb
I\omega,\omega]\Gamma\right),\\
\label{veselova}&&\dot{\mathbf e}_i=-\epsilon \omega \mathbf e_i,
\qquad\qquad\qquad\qquad\qquad
i=1,\dots,r,\\
\nonumber&&\mathbf m=\omega+\Gamma(\mathbb
I\omega-\omega)+(\mathbb I\omega-\omega)\Gamma-\Gamma(\mathbb
I\omega-\omega)\Gamma.
\end{eqnarray}

Along the trajectory $(\mathbf m(t),\mathbf e_1(t),\dots,\mathbf
e_r(t))$ of the modified Veselova problem \eqref{veselova}, the
angular velocity matrix $\omega(t)$ has the form
\begin{equation*}
\omega(t)=
\begin{pmatrix}
0 & \cdots & \omega_{1r}(t) & \cdots & \omega_{1n}(t) \\
\vdots & \ddots & \vdots &  & \vdots \\
-\omega_{1r}(t) & \cdots & 0 & \cdots & \omega_{rn}(t) \\
\vdots &  & \vdots & C_r &  \\
- \omega_{1n}(t) & \cdots & - \omega_{rn}(t) &  &
\end{pmatrix},
\end{equation*}
where $\omega_{ij}(t)=\langle \mathbf e_i(t)\wedge \mathbf
e_j(t),\omega(t)\rangle$ and $C_r$ is a constant $(n-r)\times
(n-r)$ matrix.

Next, define the left-invariant metric via the relation
\begin{equation}\label{ves-op}
\mathbb I(\mathbf E_i\wedge \mathbf E_j)=a_ia_j\mathbf E_i\wedge
\mathbf E_j,
\end{equation}
where $\mathbf E_1,\dots,\mathbf E_n$ is the standard orthonormal
base of $\R^n$ and $A=\diag(a_1,\dots,a_n)$ is positive definite.

Then,  from Theorem \ref{TLR2} and Theorem 5.1 of \cite{FeJo}, we
get:

\begin{thm}
The $\epsilon$--modified Veselova system \eqref{veselova},
\eqref{ves-op} has an invariant measure
\begin{equation*}
\left(\sum_{I} a_{i_1} \cdots a_{i_r} (\mathbf e_1 \wedge \cdots
\wedge \mathbf e_r)^2_{I}
 \right)^{(\frac{1}{2\epsilon}-1)({n-r-1})}{\mathrm d}{\mathbf m}\wedge {\mathrm
d}\mathbf e_{1} \wedge \dots \wedge {\mathrm d}\mathbf
e_r\vert_{so(n)\times V_{n,r}},
\end{equation*}
where the summation is over all $r$-tuples  $I=\{1\le i_1 < \cdots
< i_r \le n\}$, and $(\mathbf e_1 \wedge \cdots \wedge \mathbf
e_r)_{I}$ are the {\it Pl\"ucker coordinates\/} of the $r$-form
$\mathbf e_1 \wedge \cdots \wedge \mathbf e_r$. In the case $r=1$,
the density is proportional to $(\mathbf e_1, A \mathbf
e_1)^{(\frac{1}{2\epsilon}-1)({n-2})}$.
\end{thm}

Here ${\mathrm d}\mathbf e_{1}\wedge \dots \wedge {\mathrm
d}\mathbf e_r$ is the standard volume form on $\R^{nr}$. Note
that, for the inertia operator $\mathbb I$ replaced by $\mathbf
I=\mathbb I+D\mathbb E$, $r=1$ and $n=3$, the equations
\eqref{veselova} give another form of the rubber Chaplygin ball
equations \eqref{mom-eq**}.

\section{Modified L+R systems}

The {\it $\epsilon$L+R system} on the space
$$
\mathfrak g\times {\bf Sym}(\g)(\omega,\Pi) \qquad \text{or}\qquad
\mathfrak g\times {\bf Sym}(\g)(\mathbf k,\Pi)
$$
is defined by
\begin{equation}
\dot{\mathbf k}= [\mathbf k, \omega], \qquad
\dot\Pi=\epsilon[\Pi,\ad_\omega], \qquad \mathbf k=\mathbb
I\omega+\Pi\omega. \label{mL+R}
\end{equation}

\begin{thm}\label{TL+R}
The $\epsilon$L+R system possesses  an invariant measure
$$
\mu\,{\mathrm d}\omega \wedge {\mathrm d}\Pi \qquad i.e., \qquad
\mu^{-1}\,{\mathrm d}\mathbf k \wedge {\mathrm d}\Pi,
$$
 with density same as in the case of the usual L+R systems:
\begin{equation}
\mu=\sqrt{\det ({\mathbb I}+\Pi)}\, \label{meraL+R},
\end{equation} where $\mathrm d\Pi$ is the standard measure on
${\bf Sym}(\g)$.
\end{thm}

\begin{proof} The proof is a modification of the corresponding
statement for L+R systems (see \cite{FeRCD, FeKo}). We have
$$
\dot{\mathbf k}=(\mathbb I+\Pi)\dot\omega+\dot\Pi\omega=(\mathbb
I+\Pi)\dot\omega+\epsilon
\Pi[\omega,\omega]-\epsilon[\omega,\Pi\omega].
$$
Thus, we can represent the equations (\ref{mL+R}) in the form
\begin{equation}
\dot\omega=(\mathbb I+\Pi)^{-1}\left([\mathbb
I\omega,\omega]+(1-\epsilon)[\Pi\omega,\omega]\right),\qquad
\dot\Pi= \epsilon [\Pi,\ad_\omega].\label{mL+R*}
\end{equation}

The volume form \eqref{meraL+R} is invariant with respect to the
flow of \eqref{mL+R*} if and only if
\begin{equation}\label{uslov*}
\mu\left(\mathrm{div}(\dot \omega)+\mathrm{div}(\dot
\Pi)\right)+\dot\mu=0,
\end{equation}
where we take  the standard divergence on $\mathfrak g$ and ${\bf
Sym}(\g)$:
$$
\mathrm{div}(\dot\omega)=\sum_i\frac{\partial
\dot\omega_i}{\partial\omega_i}, \qquad
\mathrm{div}(\dot\Pi)=\sum_{i\le j}\frac{\partial
\dot\Pi_{ij}}{\partial \Pi_{ij}}.
$$
Here we take coordinates of $\omega$ and $\Pi$ with respect to the
orthonormal base $E_1,\dots,E_n$ of $\mathfrak g$ and the
associated base $E_i\otimes E_j$ of linear operators on $\g$.

It is clear that $\mathrm{div}(\dot \Pi)=0$. Define $n\times n$
matrix $\Omega$:
$$
\Omega_{ij}=\frac{\partial((\mathbb I+\Pi)\dot\omega)_i}{\partial
\omega_j}=\frac{\partial([\mathbb
I\omega,\omega]+(1-\epsilon)[\Pi\omega,\omega])_i}{\partial
\omega_j}, \quad i, j=1,\dots,n.
$$
Then
$$
\Omega=-\ad_\omega \circ \mathbb I+\ad_{\mathbb
I\omega}+(1-\epsilon)(-\ad_\omega \circ \Pi+\ad_{\Pi\omega})
$$
and we can write
$$
\mathrm{div}(\dot\omega)=\tr((\mathbb I+\Pi)^{-1}\Omega).
$$

In view of symmetry of $\mathbb I+\Pi$, the skew symmetric part of
$\Omega$ does not contribute to the expression for the divergence.

Taking into account \eqref{mL+R}, the symmetric part of $\Lambda$
has the form
\begin{eqnarray*}
\Omega_+ &=&\frac12(\Omega^T+\Omega)=\frac12\left( \mathbb I\circ
\ad_\omega -\ad_\omega \circ \mathbb I+ (1-\epsilon)(\Pi\circ
\ad_\omega-\ad_\omega \circ \Pi )\right)\\
&=& \frac12\left( (\mathbb I+\Pi)\circ \ad_\omega -\ad_\omega
\circ (\mathbb I+\Pi)-\dot\Pi\right).
\end{eqnarray*}

As a result,  we obtain
\begin{eqnarray*}
\mathrm{div}(\dot\omega)&=&  \tr((\mathbb
I+\Pi)^{-1}\Omega_+)\\
&=&\frac12\tr\left((\mathbb I+\Pi)^{-1}((\mathbb I+\Pi)\circ
\ad_\omega -\ad_\omega
\circ (\mathbb I+\Pi)-\dot\Pi)\right) \\
&=& -\frac12\tr\left((\mathbb
I+\Pi)^{-1}\dot\Pi\right)=-\frac12\tr\left((\mathbb
I+\Pi)^{-1}\frac{d}{dt}(\mathbb I+\Pi)\right)\\
&=&-\frac12\left(\det(\mathbb
 I +\Pi)\right)^{-1} \frac{d}{dt}\det (\mathbb I+\Pi)\\
&=& -(\sqrt{\det (\mathbb I+\Pi)})^{-1}\frac{d}{dt}\sqrt{\det
(\mathbb I+\Pi)},
\end{eqnarray*}
where we used the unimodularity condition for compact groups
$\mathrm{tr} \, \mathrm{ad}_{\omega}=0$ and the well-known
identity \eqref{ID}. We conclude that $\mu =\sqrt{\det(\mathbb
 I +\Pi)}$
satisfies the Liouville equation \eqref{uslov*} which establishes
the theorem.
\end{proof}

\begin{remark}{\rm
Note that the kinetic energy  $H=\frac 12\langle\mathbf
k,\omega\rangle$ is conserved for $\epsilon$L+R systems
\begin{eqnarray*}
\frac d{dt} \langle\mathbf k,\omega,\omega\rangle &=& 2
\langle(\mathbb I+\Pi)\dot\omega,\omega\rangle +
  \langle \frac{d}{dt}(\mathbb I+\Pi)\omega,\omega\rangle \\
& =& 2 \langle [\mathbb
I\omega,\omega]+(1-\epsilon)[\Pi\omega,\omega] ,\omega\rangle +
\epsilon\langle\Pi[\omega,\omega]-[\omega,\Pi\omega],\omega
\rangle =0,
\end{eqnarray*}
while, for the $\epsilon$LR systems, the kinetic energy $H=\frac
12\langle \mathbb I\omega,\omega\rangle$ is conserved only on the
invariant submanifold $\phi_s=\langle \omega, e_s \rangle=0$,
i.e., when $\pr_\mathcal H\omega=0$. However in the case
$\epsilon=1$, we have also the preservation of the following
modification of the kinetic energy:
$$
F=\frac12\langle \mathbb I\omega,\omega\rangle-\langle\pr_\mathcal
H\omega,\mathbb I\omega\rangle.
$$
Indeed, $\frac{d}{dt}H=\langle \omega,\Lambda\rangle$, and from
 \eqref{projektori}, \eqref{HO}, we have
\begin{eqnarray*}
\frac{d}{dt}\langle\pr_\mathcal H\omega,\mathbb I\omega\rangle &=&
\langle \epsilon \pr_\mathcal H[\omega,\omega]-\epsilon
[\omega,\pr_\mathcal H\omega], \mathbb I\omega\rangle+ \langle\pr_\mathcal H\omega,[\mathbb I\omega,\omega]+\Lambda\rangle\\
&=& (1-\epsilon)\langle\pr_\mathcal H\omega,[\mathbb
I\omega,\omega]\rangle+\langle \omega,\Lambda\rangle=\langle
\omega,\Lambda\rangle\quad \text{for}\quad \epsilon=1.
\end{eqnarray*}
}\end{remark}

\begin{remark}{\rm
The ($\epsilon$--modified) LR and L+R systems can be considered on
non-compact groups, when we have an invariant measure for
unimodular groups as well. Recall that the group $G$ is {\it
unimodular} if $\tr\ad_{\omega}=0$. }\end{remark}

\begin{example}{\rm
As in Example 1, consider the orthogonal decomposition
$\g=\mathcal H\oplus \mathcal D$, where $\mathcal H$ is the
isotropy algebra of $\gamma=e_1$. Let us take $\Pi=D\pr_\mathcal
D$. Then the modified L+R system \eqref{mL+R} takes the form
\begin{equation*}
\dot{\mathbf k}= [\mathbf k, \omega], \qquad
\dot\gamma=\epsilon[\gamma,\omega], \qquad \mathbf k=\mathbb
I\omega+D\pr_\mathcal D\omega.
\end{equation*}
After the identification $so(3)\cong \R^3$ given by \eqref{iso},
Theorem 3 recovers the invariant measure \eqref{mu-ch}. Another
natural choice for the operator $\Pi$ is
$$
\Pi= D[[\gamma,\omega],\gamma]
$$ (see \cite{H,
Jo4}), which also leads, for $\g=so(3)$, to the sphere-sphere
problem \eqref{Chap}.
 }\end{example}

\subsection{$\epsilon$L+R system on $so(n)\times V_{n,r}$}\label{L+RV}

Here we use the notation of Subsection \ref{MVS}. Consider the
decomposition of $so(n)$ given by \eqref{ort-so(n)}, and take
$\Pi=D\pr_{\mathcal D_r}$. As a result, we obtain the modified L+R
system on the space $so(n)\times V_{n,r}$
\begin{eqnarray}
\nonumber&&\dot{\mathbf k}= [\mathbf k, \omega], \\
\label{mL+R-ch}&&\dot{\mathbf e}_i=-\epsilon \omega \mathbf e_i,
\qquad \qquad\qquad\qquad\,\,\,\, i=1,\dots,r,\\
\nonumber &&\mathbf k=\mathbb
I\omega+D(\Gamma\omega+\omega\Gamma-\Gamma\omega\Gamma), \quad
\Gamma=\mathbf e_1\otimes \mathbf e_1+\dots+\mathbf e_r\otimes
\mathbf e_r.
\end{eqnarray}

For $r=1$ and $\epsilon=1$, the equations \eqref{mL+R-ch} model
the problem of rolling without slipping of a Chaplygin ball over
the hyperplane in $\R^n$, orthogonal to $\mathbf e_1$ (see
\cite{FeKo, Jo4}). As in \cite{Jo4} (see eq. (49) of \cite{Jo4}),
consider the metric defined by inertia operator
\begin{equation}\label{ch-op}
\mathbb I(\mathbf E_i\wedge\mathbf
E_j)=\frac{Da_ia_j}{D-a_ia_j}\mathbf E_i\wedge\mathbf E_j,
\end{equation}
where $0< a_ia_j <D$, $i,j=1,\dots,n$.

\begin{thm}
The $\epsilon$L+R system \eqref{mL+R-ch}, with the inertia
operator given by \eqref{ch-op}
 has an invariant measure
\begin{equation*}
\left(\sum_{I}\frac{(\mathbf e_1 \wedge \cdots \wedge \mathbf
e_r)^2_{I}}{a_{i_1} \cdots a_{i_r}}
 \right)^{-\frac12({n-r-1})}{\mathrm d}{\mathbf k}\wedge {\mathrm
d}\mathbf e_{1} \wedge \dots \wedge {\mathrm d}\mathbf
e_r\vert_{so(n)\times V_{n,r}}.
\end{equation*}
For $r=1$, the density is proportional to $(\mathbf e_1, A^{-1}
\mathbf e_1)^{-\frac{1}{2}({n-2})}$.
\end{thm}
\begin{proof}

Motivated by the relationship between 3--dimensional Veselova
problem and the Chaplygin ball problem established by Fedorov
\cite{Fe1}, define the operator $\mathbf I$ and matrixes $\mathbf
w, \mathbf m$ by
$$
\mathbf I=\mathbb E+D\mathbb I^{-1}, \qquad \mathbf w=\mathbb
I\omega, \qquad \mathbf m=\pr_{\mathcal D_r}\mathbf I\mathbf
w+\pr_{\mathcal H_r}\mathbf w.
$$
Then
$$
\mathbf m=\pr_{\mathcal D_r}\mathbb I\omega+\pr_{\mathcal
D_r}D\mathbb I^{-1}\mathbb I\omega+\pr_{\mathcal H_r}\mathbb
I\omega=\mathbb I\omega+D\pr_{\mathcal D_r}\omega=\mathbf k.
$$

Therefore, by using \eqref{det}, we get
$$
\det(\mathbb I+\Pi)=\det\frac{\partial \mathbf
k}{\partial\omega}=\det \frac{\partial \mathbf m}{\partial\mathbf
w}\cdot \det \frac{\partial \mathbf
w}{\partial\omega}=\det(\mathbf I\vert_{\mathcal
D_r})\cdot\det\mathbb I\,.
$$

Now, the definition of $\mathbb I$ can be seen as follows: it
implies the identity $\mathbf I(\mathbf E_i\wedge \mathbf
E_j)=Da_i^{-1}a_j^{-1}\mathbf E_i\wedge \mathbf E_j$. By combining
the above expressions with Theorem \ref{TL+R} and Theorem 5.1 of
\cite{FeJo} we obtain the statement.
\end{proof}

\begin{remark}{\rm
The system \eqref{Chap} is integrable by the Euler--Jacobi theorem
\cite{AKN} for $\epsilon=1$ (Chaplygin \cite{Ch1}, see also
\cite{AKN, BM}) and for $\epsilon=-1$ (Borisov and Fedorov
\cite{BF}, see also \cite{BM, BFM}). Similarly, in the case of the
rubber rolling, we have integrability for $\epsilon=1$ (see
\cite{VeVe2, EKR}) and $\epsilon=-1$ (see \cite{BM2, BM}). The
problem with the addition of potential forces is studied in
details in \cite{BMB}. On the other hand, the integrable models of
the rolling of the Chaplygin ball (on the zero level set of the
$SO(n-1)$--momentum map) and the rubber Chaplygin ball  over a
hyperplane in $\R^n$ are given in \cite{Jo4} and \cite{FeJo, Jo3},
respectively. It would be interesting to study appropriate
problems where we have a spherical surface instead of a
hyperplane. The natural $n$--dimensional variants of the equation
\eqref{Chap} and \eqref{rubber} are, respectively, the equations
\eqref{mL+R-ch} and \eqref{veselova} (with $\mathbb I$ replaces by
$\mathbf I=\mathbb I+D\mathbb E$), where we set
$r=1$.}\end{remark}

\subsection*{Acknowledgments} I am grateful to the
referee for useful suggestions.
 The research was supported by the Serbian Ministry of Education and
Science Project 174020 Geometry and Topology of Manifolds,
Classical Mechanics, and Integrable Dynamical Systems.

\end{document}